\documentclass[11pt]{article}

\usepackage{amsthm}
\usepackage{amssymb}
\usepackage{graphicx}
\usepackage{algorithm}
\usepackage{algorithmicx,algpseudocode}
\usepackage{tcolorbox}
\usepackage{geometry}
\usepackage{cleveref}
\usepackage{thm-restate}
\usepackage{mathtools}
 
\newtheorem{theorem}{Theorem}
\newtheorem{lemma}{Lemma}
\newtheorem{definition}{Definition}

\newcommand{\eps}{\epsilon}
\newcommand{\tree}{\textsc{TBAP}}
\newcommand{\follow}{\textsc{PFTAL}}
\newcommand{\submod}{\textsc{SubmodPFTAL}}
\newcommand{\bandit}{\textsc{BanditSubmodPFTAL}}

\DeclareMathOperator*{\argmin}{\arg\min}

\DeclareMathOperator{\Reg}{Regret}
\DeclareMathOperator*{\E}{\mathbb{E}}

\usepackage[hidelinks]{hyperref}
\usepackage[numbers,sort&compress]{natbib}

\usepackage{xcolor}
\definecolor{DarkGreen}{rgb}{0.1,0.5,0.1}

\newcommand{\longversion}[1]{}

\renewcommand{\longversion}[1]{}

\begin{document}

\title{Differentially Private Online Submodular Optimization}
\author{Adrian Rivera Cardoso\thanks{School of Industrial and Systems Engineering, Georgia Institute of Technology. {\tt adrian.riv@gatech.edu}. Supported in part by a TRIAD-NSF grant (award 1740776).} \and Rachel Cummings\thanks{School of Industrial and Systems Engineering, Georgia Institute of Technology. {\tt rachelc@gatech.edu}. Supported in part by a Mozilla Research Grant.}}
 \maketitle

\begin{abstract}
%
%
In this paper we develop the first algorithms for online submodular minimization that preserve differential privacy under full information feedback and bandit feedback.  A sequence of $T$ submodular functions over a collection of $n$ elements arrive online, and at each timestep the algorithm must choose a subset of $[n]$ before seeing the function. The algorithm incurs a cost equal to the function evaluated on the chosen set, and seeks to choose a sequence of sets that achieves low expected regret.

Our first result is in the full information setting, where the algorithm can observe the entire function after making its decision at each timestep.  We give an algorithm in this setting that is $\eps$-differentially private and achieves expected regret $\tilde{O}\left(\frac{n^{3/2}\sqrt{T}}{\eps}\right)$.  This algorithm works by relaxing submodular function to a convex function using the Lovasz extension, and then simulating an algorithm for differentially private online convex optimization.

Our second result is in the bandit setting, where the algorithm can only see the cost incurred by its chosen set, and does not have access to the entire function.  This setting is significantly more challenging because the algorithm does not receive enough information to compute the Lovasz extension or its subgradients.  Instead, we construct an unbiased estimate using a single-point estimation, and then simulate private online convex optimization using this estimate. Our algorithm using bandit feedback is $\eps$-differentially private and achieves expected regret $\tilde{O}\left(\frac{n^{3/2}T^{3/4}}{\eps}\right)$.
\end{abstract}


\section{Introduction}\label{s.intro}



Online learning has received significant attention due to the growing amounts of information collected about individuals, and has been studied in the context of a wide variety of optimization problems, including portfolio optimization \cite{cover1991, kalai2002, helmbold1998}, shortest paths \cite{takimoto2003path}, combinatorial optimization \cite{hazan2012submodular}, convex optimization \cite{hazan2010optimal, ben2015oracle}, and game theoretic optimization \cite{cesa2006prediction}. 
When these machine learning tools are applied to sensitive data from individuals, privacy concerns becoming increasingly important.  In applications such as clinical trials, online ad placement, personalized pricing, and recommender systems, online learning algorithms are dealing with personal (and possibly highly sensitive) data. 

In this paper, we develop the first algorithms for differentially private online submodular optimization.  A function $f: 2^{[n]} \rightarrow \mathbb{R}$ mapping from discrete collections of elements to real values is \emph{submodular} if it exhibits the following diminishing returns property: for all sets $S,S'\subseteq [n]$ such that $S'\subseteq S$ and for all elements $i\in [n] \setminus S$,
\begin{align*}
f_t(S' \cup \{i\})-f_t(S') \geq f_t(S \cup \{i\})-f_t(S).
\end{align*}
Submodular functions have several applications in machine learning (see \citep{krause2011submodularity} for a survey) and are used extensively used economics because their diminishing returns property captures preferences for substitutable goods and satiation from multiple copies of the same good \citep{topkis2011supermodularity, bach2013learning}.


In the \emph{Online Submodular Minimization} problem, a sequence of $T$ submodular functions $f_1, \ldots, f_T : 2^{[n]} \rightarrow \mathbb{R}$ arrive in an online fashion.  At every timestep $t$, a decision maker choose a set $S_t \subseteq [n]$ before observing the function $f_t$.  The decision maker then incurs cost $f_t(S_t)$.  The decision maker's goal is to minimize her total regret, which is defined as,
\begin{align*}
\Reg(T) = \sum_{t=1}^T f_t(S_t) - \min_{S\subseteq [n]} \sum_{t=1}^T f_t(S).
\end{align*}
That is, her regret is the difference between her total cost across all rounds, and the cost of the best fixed set in hindsight after seeing all the functions.  We say that an algorithm for the Online Submodular Minimization problem is \emph{no regret} is the regret (or expected regret for randomized algorithms) is sublinear in $T$: $\Reg(T) = o(T)$.

We consider two different settings based on the type of informational feedback the decision maker receives in each round.  In the \emph{full information setting}, the decision maker observes the entire function $f_t$ after making her choice of $S_t$.  In the \emph{bandit setting}, the decision maker only observes her cost $f_t(S_t)$ and does not receive any additional information about the function $f_t$.  The bandit setting is a more challenging environment because the decision maker has severely restricted information when making decisions, but also captures the reality of many real-world online learning problems where counterfactual outcomes cannot be measured.

We formally incorporate the task of preserving privacy by using the framework of differential privacy. Differential privacy was first defined by \cite{DMNS06} for algorithms operating on large static databases, and required that if a single entry in the database changed, then the algorithm would produce approximately the same output.  In this work, we view our database as the sequence of submodular functions $f_1, \ldots, f_T$, and the algorithm's output is the sequence of sets $S_1, \ldots, S_T$.  We require that if a single function $f_t$ were changed to a different $f_t'$, then the entire sequence of chosen sets would be approximately the same.  We formalize this in Definition \ref{def.dp} below.

Let $F = \{f_1, ..., f_T\}$ and $F' = \{f_1', \ldots, f_T'\}$ be sequences of functions.  We say $F$ and $F'$ are neighboring sequences if $f_t = f_t'$ for all but at most one $t \in [T]$. 
\begin{restatable}[Differential privacy \cite{DMNS06}]{definition}{dpdef}\label{def.dp}
An algorithm $\mathcal{A} : \mathcal{F}^T \rightarrow \mathcal{R}^T$ is $(\eps, \delta)$-{\em differentially private} if for all neighboring sequences $F, F' \in \mathcal{F}$ and every subset of the output space $\mathcal{S} \subseteq \mathcal{R}^T$, 
\begin{align*}
\Pr[\mathcal{A}(F) \in \mathcal{S}] \leq e^\epsilon P[\mathcal{A}(F') \in \mathcal{S}] + \delta.
\end{align*}
If $\delta=0$, we say that $\mathcal{A}$ is $\eps$-differentially private.
\end{restatable}

The main goal of this paper is to design differentially private no-regret algorithms for the Online Submodular Minimization problem.  There are many applications of online learning problems using sensitive data that could benefit from formal privacy guarantees, such as clinical drug trials, online ad placement, and personalized pricing.  For concreteness, we provide the following motivating example for the study of private online submodular optimization.

\paragraph{Motivating Example.}
As a concrete motivating example we consider the following online ad placement problem.  Online retailers such as Amazon, Walmart, and Target design their websites such that the retailers can offer other products at check out which complement the item the customer is buying. Due to item complementarities, each user $t$ has a supermodular utility function $f_t$, defined over the possible subsets of products the retailer can offer.  For the user that arrives at time $t$, the retailer must choose a set $S_t$ of products to display that maximize $f_t(S_t)$ without knowing the user's utility function.  The retailer receives bandit feedback since they can only observe $f_t(S_t)$, and not the entire function $f_t$.  The retailer seeks to minimize regret: $\max_{S\in [n]} \sum_{t=1}^T f_t(S) - \sum_{t=1}^T f_t(S_t)$. Since supermodular maximization is mathematically equivalent to submodular minimization, the retailer has to solve an online submodular minimization problem with bandit feedback. Existing product recommendation systems have been shown to leak information about users \citep{ZWJ14}, motivating the need for formal privacy guarantees in this setting.  Therefore, the retailer should perform this optimization in a differentially private manner to ensure that no individual's information is leaked to other users.

\subsection{Our Results and Techniques}
In this paper we develop the first algorithms for online submodular minimization that preserve differential privacy under full information feedback and bandit feedback.

We start with the full information setting, where the algorithm can observe the entire function $f_t$ after making its decision at each time $t$.  We give an algorithm in this setting that is both differentially private and satisfies no regret.

\begin{theorem}[Informal]
In the full information setting of Online Submodular Minimization, there is an $\eps$-differentially private algorithm that achieves regret:
\[ \E[\Reg(T)] = \tilde{O}\left(\frac{n^{3/2}\sqrt{T}}{\eps}\right). \]
\end{theorem}

This algorithm works by first relaxing each input submodular function to a convex function using the Lovasz extension (defined formally in Section \ref{s.submod}).  Our algorithm then simulates an algorithm for differentially private online convex optimization (due to \citet{ST13}) run on the sequence of Lovasz extensions.  The differential privacy guarantee is inherited from the private online convex optimization algorithm.  To prove the regret bound, we show that the relaxation and optimization on convex functions does not increase the regret guarantee by too much.  Our algorithm loses only a factor of $\sqrt{n}$ relative to the regret of \citep{ST13} for private online convex optimization. 

We next consider the bandit setting, which is significantly more challenging and requires new techniques.  The private online convex optimization algorithm of \citet{ST13} requires use of the subgradient of the Lovasz extension.  However in the bandit setting, the algorithm does not receive enough information to compute the exact Lovasz extension or its subgradients.  Instead, we construct an unbiased estimate of the subgradient using the one-point estimation method of \cite{hazan2012submodular}.  We then apply the algorithm of \cite{ST13} to the unbiased estimate of the gradient of the Lovasz extension.  This yields a differentially private no-regret algorithm for online submodular minimization in the bandit setting.

\begin{theorem}[Informal]
In the bandit setting of Online Submodular Minimization, there is an $\eps$-differentially private algorithm that achieves regret:
\[ \E[\Reg(T)] = \tilde{O}\left(\frac{n^{3/2}T^{3/4}}{\eps}\right). \]
\end{theorem}


The regret guarantees of our algorithms are worse than the best non-private algorithms by only a factor of $\sqrt{n}$ and $T^{1/12}$.

\subsection{Related Work}

Our results rely heavily on tools from \cite{ST13} and \cite{hazan2012submodular}. \cite{ST13} provides a differentially private algorithm for online convex optimization that achieves a regret rate $\tilde{O}(\frac{\sqrt{nT}}{\epsilon})$ in the full information setting, which is worse than the non-private setting by only a factor of polylog$(T) \sqrt{n}$. Under bandit feedback, they give a modification of their full information algorithm that achieves cumulative regret $\tilde{O}(\frac{nT^{3/4}}{\epsilon})$. One of the key components in our algorithms are modifications of these tools for online convex optimization, which are applied once we have relaxed the submodular functions to their convex Lovasz extensions.  \cite{hazan2012submodular} provide algorithms for non-private online submodular minimization in both the full information and bandit feedback settings. They design subgradient descent-type algorithms that achieve regret of $O(\sqrt{nT})$ and $O(nT^{2/3})$ in the full information and bandit settings respectively.  Our algorithms make use of their one-point gradient estimation technique for the bandit setting. We remark that, to the best of our knowledge, there is no known way to modify subgradient descent-type algorithms, to achieve differential privacy in the online convex bandit problem without damaging the regret bounds by less than polylog$(T)$ factors. 

Although our algorithms use these tools, composition of these previous results is not straight-forward.  The bound on the variance of the one-point gradient estimator for the Lovasz extension is not the same as that of the estimator used for online convex optimization with bandit feedback, which requires special care in the analysis. If one were to blindly compose the results of  \cite{ST13} and \cite{hazan2012submodular}, it would yield regret $O(\frac{n^2 T^{11/12}}{\eps})$ in the bandit setting, instead of the regret rate $O(\frac{n^{3/2}T^{3/4}}{\eps})$ that we achieve.


Other relevant work includes \cite{jain2012differentially}, where the authors design differentially private algorithms for online convex optimization.  However, these algorithms only achieve optimal regret rates in some special cases.   In \cite{agarwal2017price}, the authors provide differentially private algorithms for the special case of online linear optimization with bandit feedback, and obtain regret $\tilde{O}(\frac{\sqrt{T}}{\epsilon})$ which is (almost) optimal. The problem of private online submodular maximization has been studied by \cite{pmlr-v70-mitrovic17a} and \cite{gupta2010differentially}. However, our work cannot be compared to theirs since the problems of minimizing and maximizing a submodular functions are very different.  Additionally, these works only consider the offline problem with full information feedback. Finally, \cite{badanidiyuru2014streaming} studies non-private online submodular maximization only under full information feedback.

\section{Preliminaries}


In this section we present background on convex functions, submodular functions, and differential privacy that will be useful for our results in later sections.

\subsection{Convexity and Lipschitz Continuity}

For a set $X$ we define its diameter $D_X=\sup_{x,y\in X}\|x-y\|_2$.  A set $X\subseteq \mathbb{R}^d$ is a {\em convex set} if for any $x,y\in X$ and any $\lambda \in [0,1]$, $\lambda x + (1-\lambda) y \in X$.  For a function $f:X\rightarrow \mathbb{R}$, a {\em subgradient} of $f$ at a point $y$, denoted $\nabla f(y)$, is a vector $g \in \mathbb{R}^d $ such that $f(x)-f(y)\geq g^{\top}(x-y)$ for all $x \in X$.  The {\em subdifferential} of $f$ at $y$, denoted $\partial f(y)$, is the set of all subgradients of $f$ at $y$.





%

\begin{definition}[Strongly convex function]
Let $X\subseteq \mathbb{R}^d$ be a convex set.  A function $f:X\rightarrow \mathbb{R}$ is {\em $H$-strongly convex} for $H \geq 0$ if, $f(x) \geq f(y) + \nabla f(y)^\top (x-y) + \frac{H}{2}||x-y||^2_2$ for all $x,y\in X$.
If $H=0$, we say that $f$ is {\em convex}.\footnote{This is equivalent to the more commonly used definition that $f$ is convex if for any $\lambda \in[0,1]$ and for any $x,y\in X$, $\lambda f(x) + (1-\lambda) f(y)  \geq f(\lambda x + (1-\lambda) y)$.}
\end{definition}

Note that every strongly convex function is also convex.  For convex $f$, the subdifferential at every point always exists and is a closed convex set.



\begin{definition}[Lipschitz function]
A function $f:X\rightarrow \mathbb{R}$ is {\em $L$-Lipschitz continuous} with respect to a norm $||\cdot||$ if $|f(x)-f(y)| \leq L||x-y||$ for every $x,y\in X$.
\end{definition}

Lemma \ref{shalev_lipschitz} gives an equivalence between Lipschiptzness of a convex function and properties of that function's subgradients.

\begin{lemma}[\cite{shalev2012online}]\label{shalev_lipschitz}
Let $f:X\rightarrow \mathbb{R}$ be a convex function. Then $f$ is $L$-Lipschitz over $X$ with respect to norm $||\cdot||$ if and only if for all $x \in X$ and for all $\nabla f(x) \in \partial f(x)$ we have that $||\nabla f(x) ||_* \leq L$, where $||\cdot||_*$ denotes the dual norm of $||\cdot||$. 
\end{lemma}

Throughout the paper, we will say that a function $f$ is $L$-Lipschitz to indicate that $f$ is $L$-Lipschitz with respect to the $L_2$ norm $||\cdot||_2$, unless otherwise stated. We also note that the $L_2$ norm is self-dual: $(||\cdot||_{2})_* = ||\cdot||_2$ \citep{nesterov2013introductory}.

\subsection{Submodular Functions}\label{s.submod}
Submodular functions share many properties with both convex and concave functions. They can be thought of as convex functions when one is trying to minimize them, however they also exhibit a diminishing marginal returns property as some concave functions do (i.e., $f(x)=\log x$). 

\begin{definition}[Submodular function]
A function $f:2^{[n]} \rightarrow [-M, M]$ is {\em submodular} if for all sets $S,S' \subseteq [n]$ such that $S'\subseteq S$ and for all elements $i\in [n] \setminus S$,
\begin{align*}
f(S' \cup {i}) - f(S') \geq f(S \cup {i}) - f(S). 
\end{align*}
\end{definition}

The connection between convex and submodular functions is formalized through the {\em Lovasz extension} (Definition \ref{def.lovasz}), which extends a submodular function $f$ over $[n]$ to its corresponding convex function $\hat{f}$ over $[0,1]^n$.  The Lovasz extension works by first describing each point in $[0,1]^n$ as a convex combination of points in $\{0,1\}^n$, which can be interpreted as subsets of $[n]$.  It then defines $\hat{f}(x)$ as the convex combination of $f$ evaluated on the sets associated with $x$.  We first define the necessary notation.

\begin{definition}[Maximal chain \cite{hazan2012submodular}] A chain of subsets of $[n]$ is a collection of sets $A_0,...,A_p$ such that $A_0 \subset A_1 \subset \cdot \cdot \cdot \subset A_p$.   
A chain is {\em maximal} if $p=n$.  For a maximal chain, $A_0 = \emptyset$, $A_n = [n]$,  and there is a unique associated permutation $\pi : [n] \rightarrow [n]$ such that $A_{\pi (i)} = A_{\pi(i)-1} \cup \{i\}$ for all $i \in [n]$. For this permutation, we can write $A_{\pi(i)} = \{ j \in [n]: \pi(j) \leq \pi(i) \}$ for all $i \in [n]$.  
\end{definition}

Define $\mathcal{K} = [0,1]^n$. For any set $S \subseteq [n]$, let $\mathcal{X}_S \in \{0,1\}^n$ denote the {\em characteristic vector} of $S$, defined as $\mathcal{X}_S (i) = 1$ if $i\in S$ and $0$ otherwise.  For any $x \in \mathcal{K}$, there is a unique chain $A_0\subset \cdot \cdot \cdot \subset A_p$ such that $x$ can be expressed as a convex combination of the characteristic vectors of the $A_i$.  That is, $\exists \mu_1, \ldots, \mu_p > 0$ such that $x= \sum_{i=0}^p \mu_i \mathcal{X}_{A_i}$ and $\sum_{i=0}^p \mu_i = 1$. Note that if $p<n$ (i.e., the chain is not maximal), the chain can be extended to a maximal chain by setting $\mu_i = 0$ for all $i$'s corresponding the the subsets of $[n]$ that were not present in the original chain.  The chain and the weights can be found in $O(n\ln(n))$ time (see, e.g., Chap. 3 of \citet{bach2013learning}).

We are now ready to define the Lovasz extension $\hat{f}$ of submodular function $f$. 
 
 \begin{definition}[Lovasz extension]\label{def.lovasz} Let $f:2^{[n]} \rightarrow [-M, M]$ be submodular. The {\em Lovasz extension} $\hat{f}: \mathcal{K} \rightarrow [-M,M]$ of $f$ is defined as follows.  For each $x \in \mathcal{K}$, let $A_0\subset \cdot \cdot \cdot \subset A_p$ be the chain associated with $x$, and let $\mu_1, \ldots, \mu_p$ be the corresponding weights in the convex combination $x= \sum_{i=0}^p \mu_i \mathcal{X}_{A_i}$.  Then, 
 \begin{align*}
 \hat{f}(x) := \sum_{i=0}^p \mu_i f(A_i) \quad \forall x \in \mathcal{K}.
 \end{align*} 
Equivalently, the Lovasz extension can also be defined by sampling $\tau$ uniformly at random from the unit interval $[0,1]$ and considering level set $S_\tau = \{i: x(i) \geq \tau\}$.  Then $\hat{f}(x) = \mathbb{E}_\tau [f(S_\tau)]$ for each  $x \in \mathcal{K}$.
 \end{definition}

 We now provide some useful properties of the Lovasz extension.
 \begin{lemma}[\cite{Fuj05, hazan2012submodular}]\label{lemma_properties_of_extension}
  The Lovasz extension $\hat{f}$ of submodular function $f$ is convex.  Additionally, for any $x \in \mathcal{K}$, let $\emptyset = B_0 \subseteq B_1\subseteq \cdot \cdot \cdot B_n$ be any maximal chain associated with $x$ and let $\pi:[n]\rightarrow[n]$ be the corresponding permutation. Then a subgradient $g$ of $\hat{f}$ at $x$ is given by: $g(i) = f(B_{\pi(i)}) - f(B_{\pi(i) - 1})$ for all $i=1,\ldots,n$.
 \end{lemma}
 
 \begin{lemma}[\cite{JB11}]
 All subgradients $g$ of the Lovasz extension $\hat{f}:\mathcal{K}\rightarrow [-M,M]$ of a submodular function are bounded by $\|g\|_2\leq \|g\|_1\leq 4M$.
 \end{lemma}
 
 \subsection{Tools from Differential Privacy}
 
 Recall the definition of differential privacy from Section \ref{s.intro}.
 
 \dpdef*
 
The following theorem says that differential privacy is robust to \emph{post-processing}: computations performed on the output of a differentially private algorithm are still differentially private.

 \begin{theorem}[Post-processing \cite{DMNS06}]\label{thm.post}
Let $\mathcal{A} : \mathcal{D} \rightarrow \mathcal{R}$ be $(\eps, \delta)$-differentially private, and let $f : \mathcal{R} \rightarrow \mathcal{R}'$ be an arbitrary randomized function.  Then $f \circ \mathcal{A} : \mathcal{D} \rightarrow \mathcal{R}'$ is $(\eps,\delta)$-differentially private.
 \end{theorem}
 
In the remainder of this section, we review two differentially private algorithms that are needed for our results.  Section \ref{s.tbap} contains a Tree-based Aggregation Protocol (TBAP), which computes online differentially private partial sums of a stream of bits.  Section \ref{s.pftal} contains Private Follow the Approximate Leader, which is a differentially private algorithm for online convex optimization, and uses TBAP as a subroutine.

\subsubsection{Tree-Based Aggregation Protocol (TBAP)}\label{s.tbap}


The Tree-Based Aggregation Protocol is a tool for maintaining differentially private partial sums of vectors arriving in an online sequence.  At each time $t$, \tree~outputs a noisy sum of the input vectors up to time $t$.  This algorithm was first introduced by \citet{CSS11} and \citet{DNPR10}, and adapted in its current form by \citet{ST13}.

The algorithm, presented formally in Appendix \ref{app.algo}, works by maintaining a complete binary tree, where the $d$-dimensional input vectors are stored in the leaf nodes, and internal nodes in the tree store a noisy sum of all leaves in their sub-tree.  At each time $t$, \tree~receives input $z_t$ and updates the value of the $t$-th leaf node to be $z_t$.  The algorithm also updates the value of each internal node affected by this change to be the updated sum plus noise drawn according to a high-dimensional analog of Laplace noise.  The algorithm then outputs a noisy partial sum $v_t$ of the nodes in the tree that approximately sum to $z_t$.

The sum at each internal node is $(\eps/\log_2 T)$-differentially private, and by construction each $z_t$ affects only $\log_2 T$ nodes of the tree.  By the \emph{composition property} of differential privacy \citep{DMNS06}, the entire tree is $\eps$-differentially private (Theorem \ref{thm.tree_private}).

\begin{theorem}[\cite{CSS11, DNPR10}]\label{thm.tree_private}
\tree$(\{z_i\}_{i=1}^T,\mu,\eps)$ is $\epsilon$-differentially private for any $\mu>0$ and any sequence of vectors $z_1,\ldots,z_T$ that each have $L_2$ norm at most $\mu$.
\end{theorem}

In addition to being private, \tree~also provides partial sums $v_t = \sum_{i=1}^t z_t$ that are accurate (with respect to the $L_2$ norm) up to additive $O(\frac{\sqrt{d} \mu \log^2 T}{\eps})$.  This is because the $L_2$ norm of the noise at each node is Gamma distributed with standard deviation $O(\frac{\sqrt{d} \mu \log T}{\eps})$, and each partial sum is computed using at most $\log T$ nodes in the tree.

\subsubsection{Private Follow The Approximate Leader (PFTAL)}\label{s.pftal}

Private Follow The Approximate Leader (PFTAL) is an algorithm due to \citet{ST13} that takes in a sequence of strongly convex functions and outputs a sequence of points that minimizes regret. It is a variant of the Follow The Regularized Leader algorithm of \cite{hazan2007logarithmic}, with the difference that instead of using exact sums of subgradients in the update step, the algorithm uses \tree~to provide private and accurate estimates of the sums of the subgradients.  This algorithm inherits the differential privacy guarantee of \tree~via post-processing (Theorem \ref{thm.post}).  \follow~enjoys low regret due to the no-regret guarantees of Follow the Regularized Leader, and from bounds on the noise added in \tree.  The full algorithm is stated in Appendix \ref{app.algo}.

 \begin{theorem}[\citep{ST13}]\label{PFTALaccuracy}
 \follow$(\{f_i\}_{i=1}^T,H,L,X,\eps)$ is $\epsilon$-differentially private, and if $f_1, \ldots, f_T$ are $H$-strongly convex and $L$-Lipschitz, then the expected regret of \follow~satisfies:
 \begin{align*}
 \E\left[ \Reg(T) \right] = O\left(\frac{n(L+H D_X)^2 \log^{2.5}T}{\epsilon H}\right).
 \end{align*}
 

\end{theorem} 
 

\section{Full Information Setting}\label{s.fullinfo}
 
 In this section we present Submodular Private Follow The Approximate Leader (\submod) which is an algorithm for Online Submodular Minimization that is both differentially private and achieves near optimal regret.  In the full information setting, the result follows easily from \follow~applied to a modified version of the Lovasz extensions $\hat{f}_1, \ldots, \hat{f}_T$ of the input submodular functions.

The main difference between using a Follow The Approximate Leader type algorithm versus the subgradient descent type algorithm of \cite{hazan2012submodular} is the following. When using \submod~to make the decision at time $t+1$, we use all the subgradients we have observed at times $1,\ldots,t$. To contrast, if we used the algorithm of \cite{hazan2012submodular}, we would only be using the subgradient obtained at $t$. This difference is crucial when trying to incorporate privacy into the setting. 

Ideally, we would like to run \follow~on the Lovasz extensions themselves, so that we can apply the regret guarantee of Theorem \ref{PFTALaccuracy}.  However, \follow~requires strongly convex input functions, but the Lovasz extension is only guaranteed to be convex. To overcome this barrier, we regularize the Lovasz extensions to ensure strong convexity.  Define the $H$-regularized Lovasz extension as, 
\begin{equation} \label{eq_reglov}
\hat{f}^H(x) = \hat{f}(x) + \frac{H}{2} \|x\|^2. 
\end{equation}
The algorithm \submod~then runs \follow~on $\hat{f}^H_1, \ldots, \hat{f}^H_T$.


\begin{algorithm}
	\caption{Submodular Private Follow The Approximate Leader: \submod($\{f_i\}_{i=1}^T, M, H, L, [n], \eps$)}
	\begin{algorithmic}
		\State \textbf{Input:} Online sequence of submodular cost functions $\{f_1,...,f_T\}$, lower and upper bounds function values $[-M,M]$, strong convexity parameter $H$, Lipschitz parameter $L$, ground set $[n]$, privacy parameter $\epsilon$.
		\State \textbf{Output:} Sequence of sets $S_1, \ldots, S_T \subseteq [n]$ 
		\State Initialize $S_1 \leftarrow$ any subset of $[n]$
		\State Set $x_1 \leftarrow \mathcal{X}_{S_1}$   
		\State Output $S_1$
		\State Compute and pass $\nabla \hat{f}_1(x_1) + Hx_1$ to  \tree$(\{\nabla \hat{f}_{i}(x_{i}) + Hx_{i}\}, L, \eps)$, and receive current partial sum $v_1$
		\For {t=1, \ldots, T-1}
		\State $x_{t+1} \leftarrow \argmin_{x\in \mathcal{K}} \; v_t^\top x + \frac{H}{2} \sum_{j=1}^t \|x-x_j\|_2^2$
		\State Sample $\tau_{t+1} \sim U[0,1]$
		\State Output $S_{t+1} = \{i: x_{t+1}(i)>\tau_t\}$ and observe $f_{t+1}$
		\State Compute $\nabla \hat{f}_t(x_{t+1})$ and pass $\nabla \hat{f}_{t+1}(x_{t+1}) + Hx_{t+1}$ to \tree$(\{\nabla \hat{f}_{i}(x_{i}) + Hx_{i}\}, L, \eps)$,
		and receive current partial sum $v_{t+1}$
		\EndFor		
	\end{algorithmic}
\end{algorithm}

This algorithm is differentially private (Theorem \ref{thm.fullpriv}) and achieves $\tilde{O}(\sqrt{T})$ regret (Theorem \ref{thm.fullregret}). 


\begin{restatable}[Privacy guarantee]{theorem}{fullpriv}\label{thm.fullpriv}
\submod$(\{f_i\}_{i=1}^T, M, H, L, [n], \eps)$ is $\epsilon$-differentially private for any sequence of functions $f_1, \ldots,f_T$ with bounded range $[-M,M]$ and for any $M,H,L,n,T > 0$.
\end{restatable}

\begin{proof}
By Theorem \ref{thm.tree_private} we know that the output of \tree, $\{v_t\}_{t=1}^T$,  is $\epsilon$-differentially private. By Theorem \ref{thm.post} we get that the sequence $\{x_t\}_{t=1}^T$ is $\epsilon$-differentially private since the procedure $x_{t+1}\leftarrow \arg \min_{x\in K} v_t^\top x + \frac{H}{2} \sum_{j=1}^t||x-x_j||_2^2$ is simply post-processing of the $v_t$'s. Computing the output $\{S_t\}_{t=1}^T$ is further post-processing of the sequence $\{x_t\}_{t=1}^T$, and Theorem \ref{thm.post} again yields the result.  
\end{proof}

\begin{restatable}[Regret guarantee]{theorem}{fullregret}\label{thm.fullregret}
\submod$(\{f_i\}_{i=1}^T, M, H, L, [n], \eps)$ run with $H = O(\frac{M\sqrt{T}}{\sqrt{n}})$ and $L = 4M + H\sqrt{n}$ for any sequence of submodular functions $f_1, \ldots,f_T: 2^{[n]} \to [-M,M]$ for any $M,n,T>0$ guarantees,
\begin{align*}
\mathbb{E}\left[\sum_{t=1}^T f_t(S_t) - \min_{S\in [n]} \sum_{t=1}^T f_t(S)\right] \leq \tilde{O}\left(\frac{n^{3/2}M\sqrt{T} }{\epsilon}\right),
\end{align*}
where the expectation is taken over the randomness of \tree~and the sampling procedure to choose $S_t$. 
\end{restatable}

\begin{proof}
To prove the theorem, we first draw a comparison between \submod~and \follow~so that we can call upon Theorem \ref{PFTALaccuracy}.
Notice that \submod~is \follow~run on sequence of functions $\{f_t^H\}_{t=1}^T$ as defined in Equation \eqref{eq_reglov}, with two extra steps used to convert elements from $\mathcal{K}$ to subsets of $[n]$. Using the regret guarantee from \follow~(Theorem \ref{PFTALaccuracy}) on the regularized Lovasz extension we get,
\begin{equation}\label{eq.regregret}
 \mathbb{E}_{\tree}\left[\sum_{t=1}^T \hat{f}_t^H(x_t) - \min_{x\in \mathcal{K}} \sum_{t=1}^T \hat{f}_t^H(x)\right]  \leq O\left(\frac{n(L + H D_{\mathcal{K}})^2 \ln(T)^{2.5}}{\epsilon H}\right).
\end{equation}

We now transform this regret guarantee into one for the Lovasz extension. First, notice that for any $x\in \mathcal{K}$, $\|x\|\geq0$, and therefore $\sum_{t=1}^T \hat{f}_t(x_t) \leq \sum_{t=1}^T \hat{f}_t^H(x_t)$. Second, we now show that $\min_{x\in \mathcal{K}}\sum_{t=1}^T \hat{f}_t(x) \geq \min_{x\in \mathcal{K}} \sum_{t=1}^T (\hat{f}_t(x) + \frac{H}{2} ||x||^2) - \frac{HTD_{\mathcal{K}}^2}{2}$. Indeed, let $x^* = \argmin_{x\in \mathcal{K}} \sum_{t=1}^T \hat{f}_t(x)$. Then,
\begin{align*}
\min_{x\in \mathcal{K}} \sum_{t=1}^T\left(f_t(x)+ \frac{H}{2}\|x\|^2\right)  &\leq \sum_{t=1}^T \hat{f}_t(x^*) + \frac{H}{2} \|x^*\|^2\\
&= \min_{x\in \mathcal{K}} \sum_{t=1}^T \hat{f}_t (x) + \frac{H}{2}\|x^*\|^2\\
&\leq \min_{x\in \mathcal{K}} \sum_{t=1}^T \hat{f}_t (x) + \frac{H D_{\mathcal{K}}^2}{2}.
\end{align*} 

Putting this two observations together with Equation \eqref{eq.regregret} we get,
\begin{align*}
\mathbb{E}_{\tree}\left[\sum_{t=1}^T \hat{f}_t(x_t) - \min_{x\in \mathcal{K}} \sum_{t=1}^T \hat{f}_t(x)\right] &\leq \mathbb{E}_{\tree}\left[ \sum_{t=1}^T \hat{f}_t^H(x_t) - \min_{x\in \mathcal{K}} \sum_{t=1}^T \hat{f}_t^H(x) \right] + \frac{H D_{\mathcal{K}}^2}{2}\\
&\leq O\left(\frac{n(L + H D_{\mathcal{K}})^2 \ln(T)^{2.5}}{\epsilon H}\right) + \frac{H D_{\mathcal{K}}^2}{2}.
\end{align*}
Plugging $L = 4M + H D_{\mathcal{K}}$, $D_\mathcal{K}= \sqrt{n}$ and $H = O(\frac{M\sqrt{T}}{\sqrt{n}})$ yields,
\begin{align*}
\mathbb{E}_{\tree}\left[\sum_{t=1}^T \hat{f}_t(x_t) - \min_{x\in \mathcal{K}} \sum_{t=1}^T \hat{f}_t(x)\right] \leq \tilde{O}\left(\frac{n^{3/2}M\sqrt{T} }{\epsilon}\right).
\end{align*}
We are ready to conclude the proof.
\begin{align*}
\sum_{t=1}^T \mathbb{E}_{\tau, \tree}[f_t(S_t)] - \min_{S \subset [n]} \sum_{t=1}^T f_t(S) &\leq 
 \sum_{t=1}^T \mathbb{E}_{\tree}[\hat{f}_t(x_t)] -\min_{x\in \mathcal{K}}\sum_{t=1}^T \hat{f}_t(x)\\
&\leq \tilde{O}\left(\frac{n^{3/2}M\sqrt{T} }{\epsilon}\right).
\end{align*}

\end{proof}


\section{Bandit Setting}


In this section we present Submodular Private Follow The Approximate Leader with Bandit Feedback (\bandit).  This algorithm is differentially private and achieves a no regret guarantee for Online Submodular Minimization with bandit feedback.

The bandit setting makes the problem much more challenging because we do not have access to the whole function $f_t$ nor to its subgradients.  Instead we only observe the function evaluated at a single point, $f_t(S_t)$ for our chosen set $S_t$.  This means that we can no longer compute subgradients of the Lovasz extension $\nabla \hat{f}_t$ and run \follow~on the regularized $\hat{f}^H_t$ as in the full information setting.  

The key to obtain sublinear regret is to balance exploration and exploitation. In this setting, exploitation is achieved by sampling $S_t$ exactly from the distribution $\mu$ defined (through the Lovasz extension) by iterate $x_t$ of \bandit.
However, if we sample according to the distribution over sets $\mu$, we do not learn anything about the function's subgradients so, it is unclear what to do in future steps. To fix this, we should sample from some distribution that is close to $\mu$, that allows us to explore (i.e. obtain an unbiased estimate of the Lovasz extension at $x_t$). We use the sampling procedure from \citet{hazan2012submodular} to achieve this. 

With these modifications, \bandit~now works similarly to \submod~for the full information setting.  The algorithm works by computing an unbiased estimator $\hat{g}_t$ of the gradient of the Lovasz extension $\nabla \hat{f}_t$, updating a private iterate $x_t \in \mathcal{K}$ using \tree~on the regularized estimator, and outputting a random set $S_t$ that depends on $x_t$.  We now present the full algorithm of \bandit~in Algorithm \ref{alg.bandit}.
 





\begin{algorithm}
	\caption{Submodular Private Follow The Approximate Leader with Bandit Feedback: \bandit($\{f_i\}_{i=1}^T, M, H, L, [n], \eps, \gamma$)}
	\begin{algorithmic}
		\State \textbf{Input:} Online sequence of submodular cost functions $\{f_1,...,f_T\}$,  lower and upper bounds function values $[-M,M]$, strong convexity parameter $H$, Lipschitz parameter $L$, ground set $[n]$, privacy parameter $\epsilon$, parameter $\gamma$.
		\State \textbf{Output:} Sequence of sets $S_1, \ldots, S_T \subseteq [n]$
		\State Initialize $x_i \leftarrow$ arbitrary vector in $\mathcal{K}$
		\For {t=1, \ldots, T}
		\State Find maximal chain associated with $x_t$, $\emptyset = B_0 \subset B_1 \subset B_2 \subset \cdot \cdot \cdot B_n = [n]$, let $\pi$ be the associated permutation
		\State Write $x_t$ as $x_t = \sum_{i=0}^n \mu_i \mathcal{X}_{B_i}$, where $\mu_i = 0$ for the extra sets $B_i$ that where added to complete the maximal chain for $x_t$.
		\State Sample $S_t$ according to distribution: $S_t = B_i$ with probability $\rho_i = (1-\gamma)\mu_i + \frac{\gamma}{n+1}$
		\State Output $S_t$ and observe $f_t(S_t)$
		\If {$S_t = B_0$}
		\State Set $\hat{g}_t = -\frac{1}{\rho_0}f_t(B_0) e_{\pi^{-1}(1)}$
		\ElsIf {$S_t = B_n$}
		\State Set $\hat{g}_t = \frac{1}{\rho_n}f_t(B_n)e_{\pi^{-1}(n)}$
		\Else
		\State Choose $\xi \in \{+1,-1\} $ uniformly at random
		\If {$\xi = +1$}
		\State Set $\hat{g}_t = \frac{2}{\rho_i}f_t(B_i)e_{\pi^{-1}(i)}$
		\Else
		\State Set $\hat{g}_t = -\frac{2}{\rho_{i}}f_t(B_i)e_{\pi^{-1}(i+1)}$
		\EndIf
		\EndIf
		\State Pass $\hat{g}_t + H x_t$ to \tree$(\{\hat{g}_i + H x_i \}, L, \eps)$, and receive current partial sum $\hat{v}_t$
		\State Update $x_{t+1} = \argmin_{x\in \mathcal{K}} \hat{v}_t^{\top}x + \frac{H}{2}\sum_{\tau=1}^t ||x-x_\tau||^2 $
		\EndFor 
	\end{algorithmic}\label{alg.bandit}
\end{algorithm}





The analysis of \bandit~relies on the following key properties of the estimate $\hat{g}$.\footnote{Our Lemmas \ref{lem.unbiased} and \ref{lem.bound} were asserted without proof in \cite{hazan2012submodular}.  Due to minor errors in the construction of $\hat{g_t}$ in \cite{hazan2012submodular}, these claims are easily seen to be false under their construction.  Here, we build the correct estimator and prove its correctness.} Proofs are deferred to the Appendix.

\begin{restatable}{lemma}{unbiased}\label{lem.unbiased}
The random vector $\hat{g}_t$ computed in \bandit~is an unbiased estimate of a subgradient of the Lovasz extension $\hat{f}_t$ of submodular $f_t$, evaluated at point $x_t$.  That is, $$\mathbb{E}\left[\hat{g}_t \mid x_t\right] = \nabla \hat{f}_t(x_t).$$
\end{restatable}

\begin{restatable}{lemma}{bound}\label{lem.bound}
The random vector $\hat{g}_t$ computed in \bandit~satisfies the following bound on its expected $L_2$-norm,
\begin{align*}
\mathbb{E}\left[\|\hat{g}_t\|^2\right] \leq \frac{16M^2n^2}{\gamma},
\end{align*}
where the expectation is taken over the algorithm's internal randomness up to time $t$.
\end{restatable}



The exploration-exploitation dilemma can be better understood through the parameter $\gamma$.  This parameter trades off between variance of the estimate $\hat{g}_t$ and the approximation of the Lovasz extension $\hat{f}_t$ to the true submodular function $f_t$.  When $\gamma$ is large, the variance of $\hat{g}_t$ is diminished, as can be seen in the statement of Lemma \ref{lem.bound}.  When $\gamma$ is small, the performance of $f_t(S_t)$ is close to that of $\hat{f}_t(x_t)$ (see Lemma \ref{continous_to_set} in Section \ref{s.regret}). In the statement of our main result (Theorem \ref{thm:DP_bandit}), we optimally tune $\gamma$ to balance exploration and exploitation and minimize overall regret of \bandit.

 Our two main results of this section show that \bandit~is differentially private and achieves low regret.

 \begin{restatable}[Privacy guarantee]{theorem}{banditpriv}\label{thm:privacy_DP_bandit}
\bandit$(\{f_i\}_{i=1}^T, M, H, L, [n], \eps, \gamma)$ is $\epsilon$-differentially private for any sequence of functions $f_1, \ldots,f_T$ with bounded range $[-M,M]$ and for any $M,H,L,n,T,\gamma > 0$.
\end{restatable}

\begin{proof} 
By Theorem \ref{thm.tree_private} we know that the output of \tree, $\{v_t\}_{t=1}^T$,  is $\epsilon$-differentially private. Notice that \bandit~is running \follow~on regularized functions $\hat{g}_t^\top x + \frac{H}{2}\|x\|^2$ thus by the same reasoning as in Theorem \ref{thm.fullpriv}, the sequence $\{x_t\}_{t=1}^T$ is $\epsilon$-differentially private since the procedure $x_{t+1}\leftarrow \arg \min_{x\in K} v_t^\top x + \frac{H}{2} \sum_{j=1}^t \|x-x_j\|_2^2$ is simply post-processing of the $v_t$'s. Since $\{S_t\}_{t=1}^T$ is post-processing on the sequence $\{x_t\}_{t=1}^T$, applying Theorem \ref{thm.post} again completes the proof.
\end{proof}
 
 \begin{restatable}[Regret guarantee]{theorem}{banditregret}\label{thm:DP_bandit}
 \bandit$(\{f_i\}_{i=1}^T, M, H, L, [n], \eps, \gamma)$ run with $H=O(\frac{M}{\sqrt{n}T^{1/4}})$, $L=4M+2H\sqrt{n}$, and $\gamma = \frac{n}{T^{1/4}}$ for any sequence of submodular functions $f_1, \ldots,f_T: 2^{[n]} \to [-M,M]$ for any $M,n,T>0$ guarantees, 
 \begin{align*}
 \mathbb{E}\left[\sum_{t=1}^T f_t(S_t) - \min_{S \subseteq [n]} \sum_{t=1}^T f_t(S)\right] \leq \tilde{O}\left(\frac{M n^{3/2}T^{3/4}}{\epsilon}\right). 
 \end{align*}
 \end{restatable}

The proof of Theorem \ref{thm:DP_bandit} relies on several key lemmas, presented in Section \ref{s.regret}.

\subsection{Regret Analysis of \bandit}\label{s.regret}



There are several sources of potential sub-optimality in the output of \bandit~that must be bounded.  Firstly, the algorithm optimizes using continuous iterates $x_t$ instead of discrete (Lemma \ref{continous_to_set}).  Secondly, it uses the $H$-regularized Lovasz extension instead of the true Lovasz extension to compute iterates (Lemma \ref{lovasz_to_regularized_lovasz}). The algorithm incurs additional loss from the noise added in \tree~to preserve privacy (Lemma \ref{h_to_tilde_h}). Due to the bandit feedback, we cannot compute an exact subgradient of the regularized Lovasz extension, and must instead use a (random) unbiased estimator (Lemma \ref{lemma_h_to_hat_h}).

The following lemmas bound the regret from these sources of error, and are used in the proof of Theorem \ref{thm:DP_bandit} presented at the end of the section.  All omitted proofs are presented in the appendix.







We start with a lemma from \citet{hazan2012submodular}, showing that the additional loss from choosing a subset of the ground set $S_t$ instead of the point in $x_t \in \mathcal{K}$ is not too large.


\begin{restatable}[\cite{hazan2012submodular}]{lemma}{cont}\label{continous_to_set}
For any submodular function $f_t:[n] \rightarrow [-M,M]$, let $x_t$ and $S_t$ be the corresponding iterates and sets as defined in \bandit, then $\mathbb{E}[f_t(S_t)] \leq \mathbb{E}[\hat{f}_t(x_t)] + 2\gamma M$.
\end{restatable}

%

As in Section \ref{s.fullinfo}, the regret guarantees of \follow~require input functions that are strongly convex, but the Lovasz extension $\hat{f}$ of submodular $f$ is only convex.  We again regularize the Lovasz extension to ensure that it is strongly convex. Recall the regularized Lovasz extension, as defined in Equation \ref{eq_reglov}: 
\begin{equation*}
\hat{f}^H(x) = \hat{f}(x) + \frac{H}{2} ||x||^2. 
\end{equation*} 
Recall also that $\hat{f}^H_t$ is $H$-strongly convex, satisfies $\nabla \hat{f}^H_t (x) = \nabla \hat{f}_t(x) + Hx$, and is $(4M + H D_{\mathcal{K}})$-Lipschitz continuous. Since $\hat{g}_t$ is an unbiased estimate of the subgradient of the Lovasz extension at point $x$ (i.e., $\mathbb{E}[\hat{g}_t | x] = \nabla \hat{f}_t(x)$ by Lemma \ref{lem.unbiased}), then $\mathbb{E}[ \hat{g}_t + H x] = \nabla \hat{f}^H_t(x)$. 


We now show that the additional regret from regularized Lovasz extension instead of the Lovasz extension is not too high. The following lemma was stated without proof in \cite{ST13}; we provide a proof in the appendix for completeness.

\begin{restatable}[\cite{ST13}]{lemma}{lovasz}\label{lovasz_to_regularized_lovasz}
Let $\{f_t\}_{t=1}^T$ be any sequence of submodular functions, let $\{\hat{f}_t\}_{t=1}^T$ be their Lovasz extensions, let $\{\hat{f}^H_t\}_{t=1}^T$ be their regularized Lovasz extensions, let $\{x_t\}_{t=1}^T$ be any sequence of elements in $\mathcal{K}$. It holds that
\begin{align*}
\sum_{t=1}^T \hat{f}_t(x_t) - \min_{x\in \mathcal{K}} \sum_{t=1}^T \hat{f}_t (x) \leq\sum_{t=1}^T \hat{f}^H_t(x_t) - \min_{x\in \mathcal{K}} \sum_{t=1}^T \hat{f}^H_t (x) + \frac{TH}{2}D^2_{\mathcal{K}}.
\end{align*}
\end{restatable}


It will be useful in our analysis to define $h_t(x)$, which is a quadratic lower bound on the regularized Lovasz extension $\hat{f}^H_t (x)$ since the regularized Lovasz extension is $H$-strongly convex:
\begin{equation}\label{eq.lowerbound}
h_t(x) =\hat{f}^H_t(x_t) + \nabla \hat{f}^H_t(x_t)^{\top} (x-x_t) + \frac{H}{2} ||x-x_t||^2.
\end{equation}
Note that $h_t(x)$ is $(4M + H D_{\mathcal{K}})$-Lipschitz continuous. Indeed $||\nabla h_t(x)||=||\nabla\hat{f}_t^H(x) +H(x-x_t)||\leq 4M+ H D_\mathcal{K}$. 

Our next lemma shows that analyzing this lower bound instead of the regularized Lovasz extension does not harm regret by too much.

\begin{restatable}{lemma}{regular}\label{regularized_lovasz_to_h}
Let $\{f_t\}_{t=1}^T$ be any sequence of submodular functions, let $\{\hat{f}^H_t\}_{t=1}^T$ be their regularized Lovasz extensions, let $\{x_t\}_{t=1}^T$ be any sequence of elements in $\mathcal{K}$. It holds that
\begin{align*}
\sum_{t=1}^T \hat{f}^H_t(x_t) - \min_{x\in \mathcal{K}} \sum_{t=1}^T \hat{f}^H_t (x) \leq \sum_{t=1}^T h_t(x_t) - \min_{x\in \mathcal{K}} \sum_{t=1}^T h_t(x).
\end{align*} 
\end{restatable}




For our analysis, we introduce random functions $\tilde{h}_t$ that satisfy $\mathbb{E}[\tilde{h}_t(x)] = h_t(x)$ for all $x \in \mathcal{K}$.  Define $\tilde{h}_t$ as follows:
\begin{align*}
\tilde{h}_t(x) = \hat{f}^H_t(x_t) - \nabla \hat{f}^H_t (x_t)^\top x_t + (\hat{g}_t + H x_t)^\top x + \frac{H}{2} ||x-x_t||^2.
\end{align*}
The function $\tilde{h}_t$ is $(\frac{2M(n+1)}{\gamma} + H D_{\mathcal{K}})$-Lipschitz continuous because $\|\nabla \tilde{h}_t \| = \| \hat{g}_t + H x_t + H (x-x_t) \| = \|\frac{1}{\rho}f(S) + H x\| \leq \frac{2M(n+1)}{\gamma} + H D_{\mathcal{K}}$.

If we were in a non-private setting, we would define the update step to $x_{t+1}$ in \bandit~as,
\begin{align*}
x_{t+1} := \arg \min_{x\in \mathcal{K}} \sum_{\tau=1}^t \tilde{h}_{\tau} ( x) = \arg \min_{x\in \mathcal{K}} \sum_{\tau = 1}^t (\hat{g}_t+H x_\tau)^\top x + \frac{H}{2}\sum_{\tau = 1}^t \|x-x_\tau\|^2_2,  
\end{align*} 
where the second equality holds since the first two terms that define $\tilde{h}_t$ do not contain $x$. However, since we desire a differentially private algorithm, we will instead use the private partial sum $\hat{v}_t$ from \tree~to approximate $\sum_{\tau=1}^t \hat{g}_\tau + H x_\tau$.  Thus the private update is, 
 \begin{align*}
 x_{t+1} = \arg \min_{x\in \mathcal{K}} \hat{v}_{t+1}^\top x+ \frac{H}{2} \sum_{\tau=1}^t \|x-x_\tau\|^2_2.
 \end{align*}

Lemma \ref{h_to_tilde_h}, due to \cite{ST13}, bounds the additional regret from using the private updates $x_t$.  Lemma \ref{lemma_h_to_hat_h}, bounds the additional regret from using $\tilde{h}$ instead of $h$.  Lemma \ref{lemma_h_to_hat_h} closely parallels Lemma 15 of \cite{ST13}, although we achieve a tighter bound that scales as $\Theta(T^{1/2}/\gamma^{1/2})$, compared to their bound that scales as $\Theta(T^{1/2}/\gamma)$.




\begin{restatable}[\cite{ST13}, Lemma 8]{lemma}{tildeh}\label{h_to_tilde_h}
\begin{align*}
\mathbb{E}\left[\sum_{t=1}^T h_t (x_t) - \min_{x\in \mathcal{K}} \sum_{t=1}^T h(x)\right] \leq \mathbb{E}\left[\sum_{t=1}^T h_t (\hat{x}_t) -  \min_{x\in \mathcal{K}} \sum_{t=1}^T h(x) \right] + \frac{8n(4M+2HD_{\mathcal{K} })^2 \ln^{2.5}T}{H \epsilon},
\end{align*}
where the expectation is taken over all the randomness of \bandit. 
\end{restatable}

\begin{restatable}{lemma}{hath}\label{lemma_h_to_hat_h}
Conditioning on the internal randomness of \tree, 
\begin{align*}
\mathbb{E}\left[\sum_{t=1}^T h_t(\hat{x}_t) - \min_{x\in \mathcal{K}} \sum_{t=1}^T h_t(x)\right] \leq \mathbb{E}\left[\sum_{t=1}^T \tilde{h}_t(\hat{x}_t) - \min_{x\in \mathcal{K}} \sum_{t=1}^T \tilde{h}_t(x)\right] + \frac{8MnD_{\mathcal{K}} \sqrt{T}}{\sqrt{\gamma}},
\end{align*}
even when the sequence of functions $\{f_t\}_{t=1}^T$ is chosen by an adaptive adversary. Here the expectation is taken over the randomness of \submod~used to build estimates of the subgradient $\{\hat{g}_t\}_{t=1}^T$. 
\end{restatable}

Our final lemma comes from \citet{ST13}, and bounds the regret of $\tilde{h}_t$ on non-private iterates $\hat{x}_t$.

\begin{lemma}[\citep{ST13}]\label{lem.htilderegret}
Follow The Approximate Leader run on $H$-strongly convex and $L$-Lipschitz functions $\{f_t\}_{t=1}^T$ guaranteesm
\begin{align*}
\sum_{t=1}^T f_t(\hat{x}_t) - \min_{x\in X} \sum_{t=1}^T f_t(x) \leq \frac{2(L+H D_X)^2 \ln(T)}{H},
\end{align*}
where $D_X$ is the diameter of the set $X$.
\end{lemma}

We are now ready to prove Theorem \ref{thm:DP_bandit}.  We restate the theorem here for convenience.

\banditregret*

\begin{proof}[Proof of Theorem \ref{thm:DP_bandit}]
\begin{align*}
&\mathbb{E}\left[\sum_{t=1}^T f_t(S_t)\right] - \min_{S\subseteq [n]} \sum_{t=1}^T f_t(S)\\
&\leq \mathbb{E}\left[\sum_{t=1}^T f_t(S_t)\right] - \min_{x\in \mathcal{K}}\sum_{t=1}^T \hat{f}_t(x) \\
&\leq \mathbb{E}\left[\sum_{t=1}^T \hat{f}_t(x_t)\right] -  \min_{x\in \mathcal{K}}\sum_{t=1}^T \hat{f}_t(x) + 2\gamma M T \quad \text{(by Lemma \ref{continous_to_set})}\\
&\leq \mathbb{E}\left[\sum_{t=1}^T \hat{f}^H_t(x_t)\right] -  \min_{x\in \mathcal{K}}\sum_{t=1}^T \hat{f}^H_t(x) + 2\gamma M T + \frac{THD^2_\mathcal{K}}{2} \quad \text{(by Lemma \ref{lovasz_to_regularized_lovasz})}\\
&\leq \mathbb{E}\left[\sum_{t=1}^T h_t(x_t)\right] -  \min_{x\in \mathcal{K}}\sum_{t=1}^T h_t(x) + 2\gamma M T + \frac{THD^2_\mathcal{K}}{2} \quad \text{(by Lemma \ref{regularized_lovasz_to_h})}\\
&\leq \mathbb{E}\left[\sum_{t=1}^T h_t(\hat{x}_t) -  \min_{x\in \mathcal{K}}\sum_{t=1}^T h_t(x)\right] + 2\gamma M T + \frac{THD^2_\mathcal{K}}{2} + \frac{8n(4M + 2HD_{\mathcal{K}})^2 \ln^{2.5}T}{H\epsilon} \quad \text{(by Lemma \ref{h_to_tilde_h}) }\\
&\leq \mathbb{E}\left[\sum_{t=1}^T \tilde{h}_t(\hat{x}_t) -  \min_{x\in \mathcal{K}}\sum_{t=1}^T \tilde{h}_t(x)\right] + 2\gamma M T + \frac{THD^2_\mathcal{K}}{2} + \frac{8n(4M + 2HD_{\mathcal{K}})^2 \ln^{2.5}T}{H\epsilon} \\ & \quad \quad \quad \quad + \frac{8nMD_{\mathcal{K}} \sqrt{T}}{\sqrt{\gamma}} \quad \text{(by Lemma \ref{lemma_h_to_hat_h}) }\\
%
&\leq \frac{2(\frac{2M(n+1)}{\gamma} + 2HD_{\mathcal{K}})^2 \ln T}{H} + 2\gamma M T + \frac{THD^2_\mathcal{K}}{2} +\frac{8n(4M + 2HD_{\mathcal{K}})^2 \ln^{2.5}T}{H\epsilon} \\ & \quad \quad + \frac{8nMD_{\mathcal{K}} \sqrt{T}}{\sqrt{\gamma}} \quad \text{(by Lemma \ref{lem.htilderegret}) }\\
&= \tilde{O}\left(\frac{n^{3/2}M T^{3/4}}{\epsilon}\right)
\end{align*}
where the last line comes from our choice of $\gamma = \frac{n}{T^{1/4}}$ and $H = \frac{M}{\sqrt{n}T^{1/4}}$ as in the theorem statement. 
\end{proof}


\bibliography{mybib}
\bibliographystyle{abbrvnat}

\appendix


\section{Algorithms from Preliminaries}\label{app.algo}

\begin{algorithm}[h!]
	\caption{Tree Based Aggregation Protocol: \tree$(\{z_i\}_{i=1}^T, \mu, \eps)$} 
	\begin{algorithmic}
	\State \textbf{Input:} Online sequence of vectors $z_1,..., z_T \in \mathbb{R}^d$, $\mu: L_2-$norm bound on each $z_i$, privacy parameter $\epsilon$.
	\State \textbf{Output:} Sequence of noisy partial sums $v_1, \ldots, v_n \in \mathbb{R}$
	\end{algorithmic}
	\begin{algorithmic}
		\State Initialize a binary tree $A$ of size $2^{\lceil \log_2 T\rceil + 1}-1$ with leaves $z_1,...,z_T$ 
		\For {t = 1, \ldots, T}
		\State Accept $z_t$ from the data stream.
		\State Let $P = \{z_t \rightarrow \cdot \cdot \cdot \rightarrow root\}$ be the path from $z_t$ to the root. 
		\Procedure {Tree update}{}
		\State Let $\Lambda$ be the first node in $P$ that is left-child in $A$.  Let $P_\Lambda = \{z_t \rightarrow \cdot \cdot \cdot \rightarrow \Lambda\}$.
		\For {\textbf{all} nodes $\alpha$ in path $P$}
			\State $\alpha \leftarrow \alpha+ z_t$
			\If{$\alpha \in P_\Lambda$}{ $\alpha \leftarrow \alpha + \gamma$ where $\gamma \in \mathbb{R}^d$ is sampled by $\Pr[\gamma = \hat{\gamma}]\propto e^{-\frac{\|\hat{\gamma}\|_2 \epsilon}{\mu(\lceil \log_2 T\rceil + 1)}}$} 
			\EndIf 
		\EndFor
		\EndProcedure
		\Procedure {Output private partial sum}{}
		\State Initialize vector $v_t\in \mathbb{R}^d$ to zero. Let $b$ be a $(\lceil \log_2T\rceil + 1)$-bit binary representation of $t$.
		\For {$i = 1, \ldots, [\log_2T + 1] $}
			\If{ bit $b_i=1$}
			 	\If {$i$-th node in $P$ (denoted $P(i)$) is the left child in A,}{ $v\leftarrow v + P(i)$}
				\Else{ $v_t \leftarrow v_t + $left sibling $P(i)$} 
				\EndIf
			\EndIf
		\EndFor
		\State \textbf{return} noisy partial sum $v_t$
		\EndProcedure
		\EndFor
	\end{algorithmic}
\end{algorithm}

\begin{algorithm}[h!]
	\caption{Private Follow The Approximate Leader: \follow($\{f_i\}_{i=1}^T, H, L, X, \eps$)}
	\begin{algorithmic} \label{algo.ftal}
		\State \textbf{Input:} Online sequence of strongly convex cost functions $\{f_1,...,f_T\}$, strong convexity parameter $H$, Lipschitz parameter $L$, convex and compact set $X \subset \mathbb{R}^n$, privacy parameter $\epsilon$.
		\State \textbf{Output:} Sequence of vectors $x_1, \ldots, x_T \in X$
		\State Initialize $x_1 \leftarrow$ any vector from $X$
		\State Output $x_1$
		\State Compute and pass $\nabla f_1(x_1)$ into \tree$(\{\nabla f_i(x_i)\}, L, \eps)$, and receive current partial sum $v_1$
		\For {t=1, \ldots, T-1}
		\State $x_{t+1} \leftarrow \argmin_{x\in X} \; v_t^\top x + \frac{H}{2} \sum_{j=1}^t \|x-x_j\|_2^2$
		\State Output $x_{t+1}$ and observe $f_{t+1}$
		\State Compute and pass $\nabla f_{t+1}(x_{t+1})$ into \tree$(\{\nabla f_i(x_i)\}, L, \eps)$, and receive current partial sum $v_{t+1}$
		\EndFor		
	\end{algorithmic}
\end{algorithm}

\section{Omitted proofs}

\unbiased*

\begin{proof}
Notice that conditioned on the randomness up to $t-1$
\begin{equation}
\hat{g}_t = \begin{cases} 
-\frac{1}{\rho_0}f_t(B_{0})e(\pi^{-1}(1)) &\  \text{with probability $\rho_0$}\\
\frac{2}{\rho_i}f_t(B_{i})e(\pi^{-1}(i)) &\ \text{with probability $\frac{\rho_i}{2}$ for $1\leq i \leq n-1$}\\
-\frac{2}{\rho_i}f_t(B_{i})e(\pi^{-1}(i+1)) &\ \text{with probability $\frac{\rho_i}{2}$ for $1\leq i \leq n-1$}\\
\frac{1}{\rho_n}f_t(B_n)e({\pi^{-1}(n)}) &\ \text{with probability $\rho_n$}
\end{cases}
\end{equation}
Therefore
\begin{align*}
\mathbb{E}_t[\hat{g}_t] &= \rho_0[ -\frac{1}{\rho_0}f_t(B_0)e(\pi^{-1}(1)) ] + \frac{\rho_1}{2}[\frac{2}{\rho_1}f_t(B_{1})e(\pi^{-1}(1))-\frac{2}{\rho_1}f_t(B_{1})e(\pi^{-1}(2))]\\
&+ ... + \frac{\rho_{n-1}}{2} [\frac{2}{\rho_{n-1}}f_t(B_{
n-1})e(\pi^{-1}(n-1)) -\frac{2}{\rho_{n-1}}f_t(B_{n-1})e(\pi^{-1}(n)) ] + \rho_n [\frac{1}{\rho_n}f_t(B_n) e(\pi^{-1}(n))] \\
& = [f_t(B_{1}) - f_t(B_{0})] e(\pi^{-1}(1)) + [f_t(B_{2}) - f_t(B_{1})] e(\pi^{-1}(2))+ ... + [f_t(B_{n}) - f_t(B_{n-1})] e(\pi^{-1}(n)) 
 \end{align*}
This means that $\mathbb{E}_t[\hat{g}_t](\pi^{-1}(i)) = f(B_i) - f_t(B_{i-1})$ for $i=1,...,n$.
This concludes the proof since $\mathbb{E}_t[\hat{g}_t](i) = \mathbb{E}_t[\hat{g}_t](\pi^{-1}[\pi(i)]) = f_t(B_{\pi(i)}) - f_t(B_{\pi(i)-1}) = g_t(i)$ for $i=1,...,n$. 
\end{proof}

\bound*

\begin{proof}
\begin{align*}
\mathbb{E}_t [||\hat{g}_t||^2] &= \rho_0 [-\frac{1}{\rho_0}f_t(B_0)]^2 + \sum_{i=1}^{n-1} \frac{\rho_i}{2} [(\frac{2}{\rho_i}f_t(B_i))^2 + (-\frac{2}{\rho_i}) f_t(B_i)^2] + \rho_n [\frac{1}{\rho_n}f_t(B_n)^2]\\
&\leq 4 M^2 \sum_{i=0}^n \frac{1}{\rho_i}\\
&= 4 M^2 \sum_{i=0}^n \frac{1}{(1-\gamma)\mu_i + \gamma/(n+1)}\\
&= \sum_{i=0}^n \frac{n+1}{(1-\gamma)\mu_i (n+1) + \gamma}\\
&\leq \frac{4M^2(n+1)^2}{\gamma}\\
&\leq \frac{16M^2n^2}{\gamma}
\end{align*}
The second to last inequality holds as long as $\gamma\leq1$ which will be ensured by our choice of parameters of the algorithm.
\end{proof}

\cont*

The proof is identical to that of \citet{hazan2012submodular}.  We present it here for completeness.
\begin{proof}
We know $\mathbb{E}_t[f_t(S_t)] = \sum_{i=0}^n \rho_i f_t(B_i)$ and $\hat{f}_t(x_t) = \sum_{i=0}^n \mu_i f(B_i)$. Therefore,
\begin{align*}
\mathbb{E}_t [f_t (S_t)] - \hat{f}_t(x_t) &= \sum_{i=0}^n (\rho_i - \mu_i) f_t(B_i)\\
& \leq \gamma \sum_{i=0}^n \left[\frac{1}{n+1}+\mu_i \right] \left|f_t (B_i)\right| \\
& = \gamma \left( \frac{n}{n+1} + 1 \right) M \\
& \leq 2 \gamma M.
\end{align*}

Taking expectation with respect to the randomness up to time $t-1$ yields the result.
\end{proof}

\lovasz*

Lemma \ref{lovasz_to_regularized_lovasz} was stated without proof in \cite{ST13}; we provide a proof here for completeness.

\begin{proof}
Define $\bar{x} := \arg \min \sum_{t=1}^T \hat{f}^H_t(x) - \frac{H}{2} || x||^2 $. By the definition of $\hat{f}_t^H(x)$,
\begin{align*}
 \sum_{t=1}^T \hat{f}_t(x_t) - \min_{x\in \mathcal{K}} \sum_{t=1}^T \hat{f}_t (x) &= \sum_{t=1}^T \hat{f}^H_t(x_t) - \frac{H}{2} || x_t||^2 - \min_{x \in \mathcal{K}} \left\{\sum_{t=1}^T \hat{f}^H_t(x) - \frac{H}{2} ||x||^2\right\}\\
& \leq \sum_{t=1}^T \hat{f}^H_t(x_t)  - \min_{x \in \mathcal{K}} \left\{\sum_{t=1}^T \hat{f}^H_t(x) - \frac{H}{2} ||x||^2\right\}\\
& = \sum_{t=1}^T \hat{f}^H_t(x_t)  - \sum_{t=1}^T \hat{f}^H_t(\bar{x}) + \frac{TH}{2} ||\bar{x}||^2\\
&\leq \sum_{t=1}^T \hat{f}^H_t(x_t)  - \min_{x \in \mathcal{K} }\sum_{t=1}^T \hat{f}^H_t(x) + \frac{TH}{2} D_{\mathcal{K}}^2.
\end{align*}
\end{proof}

\regular*

\begin{proof}
By definition of $h_t$ we have $h_t(x_t) = \hat{f}^H_t (x_t)$. Since $h_t(x)$ is a lower bound on $\hat{f}^H_t(x)$ it follows that $\min_{x\in \mathcal{K}} \sum_{t=1}^T h_t(x) \leq \min_{x\in \mathcal{K}} \sum_{t=1}^T \hat{f}^H_t(x)$ which yields the result. 
\end{proof}

\tildeh*

We provide a proof of Lemma \ref{h_to_tilde_h} in our own notation for completeness.

\begin{proof}
Let $J_t(x) = (\sum_{\tau=1}^t\hat{g}_t)^\top x + \frac{H}{2} \sum_{\tau=1}^t \|x-x_\tau\|^2$ and let $\zeta_t \in \mathbb{R}^d$ denote the random noise added by \tree~at time $t$.  That is, $\zeta_t = g_t - \hat{g}_t$.  Then we can write $x_{t+1} = \arg \min_{x \in \mathcal{K}} J_t(x)$ and $x_{t+1} = \min_{x\in \mathcal{K}} J_t(x) + \zeta_t^\top x$.


Since $J_t(x)$ is $Ht$ strongly convex we have 
\begin{align*}
\| x_{t+1} - x_{t+1}\| \leq \frac{2 \|\zeta_t\|}{Ht}.
\end{align*}
Since $h_t$ is $(4M+2HD_{\mathcal{K}})$-Lipschitz the expression above implies,
\begin{align*}
\left|h_t(x_t) - h_t(\hat{x}_t)\right| \leq \frac{2 (4M + 2HD_{\mathcal{K}})\|\zeta_{t-1}\|}{Ht}.
\end{align*}
It follows that, 
\begin{align*}
\sum_{t=1}^T h_t(x_t) - \min_{x\in \mathcal{K}} \sum_{t=1}^T h(x) &\leq \sum_{t=1}^T h_t(\hat{x}_t) - \min_{x\in \mathcal{K}} \sum_{t=1}^T h(x) + \frac{2(4M+2HD_{\mathcal{K}})}{H} \sum_{t=1}^T \frac{\|\zeta_{t-1}\|}{t}.
\end{align*}

Each $\|\zeta_t\|$ is formed in \tree~by adding at most $\lceil \ln (T)\rceil + 1$ vectors with norms drawn from a Gamma distribution with scale $n$ and shape $\frac{(\lceil \ln (T)\rceil + 1)(L +HD_{\mathcal{K}})}{\epsilon}$, we can bound $\mathbb{E}[\|\zeta_t\|] \leq \frac{4 n \ln^{1.5}T (L + HD_{\mathcal{K}})}{\epsilon}$ for all $t$.  Combining this with the fact that $\sum_{t=1}^T \frac{1}{t} \leq \ln(T)$, we can write,

\begin{align*}
&\mathbb{E}_{\{\zeta_t\}_{t=1}^T}\left[\sum_{t=1}^T h_t(x_t) - \min_{x\in \mathcal{K}} \sum_{t=1}^T h(x) \mid \{\xi_t, S_t\}_{t=1}^T \right]\\
&\leq \mathbb{E}_{\{\zeta_t\}_{t=1}^T}\left[\sum_{t=1}^T h_t(\hat{x}_t) - \min_{x\in \mathcal{K}} \sum_{t=1}^T h(x) \mid \{\xi_t, S_t\}_{t=1}^T \right] + \frac{8n(4M + 2HD_{\mathcal{K}})^2 \ln^{2.5}T}{H\epsilon}. 
\end{align*}

Now taking expectation over all the algorithm's randomness,
\begin{align*}
&\mathbb{E}_{\{\zeta_t\}_{t=1}^T, \{\xi_t,S_t\}_{t=1}^T} \left[ \sum_{t=1}^T h(x_t) - \min_{x\in \mathcal{K}} \sum_{t=1}^T h(x) \right] \\ 
& \leq \mathbb{E}_{\{\xi_t, S_t\}_{t=1}^T, \{\zeta_t\}_{t=1}^T} \left[ \sum_{t=1}^T h_t(\hat{x}_t) - \min_{x \in \mathcal{K}}\sum_{t=1}^T h_t(x)\right] + \frac{8n(4M + 2HD)^2 \ln^{2.5}T}{H\epsilon}. 
\end{align*}
\end{proof}

\hath*

\begin{proof}

Let $x^* = \argmin_{x\in \mathcal{K}}\sum_{t=1}^T h_t(x)$. Then,
\begin{align}\label{eq.hbound}
\mathbb{E}\left[\sum_{t=1}^T h_t(\hat{x}_t) - \min_{x\in \mathcal{K}}\sum_{t=1}^T h_t(x)\right]  &= \mathbb{E}\left[\sum_{t=1}^T h_t(\hat{x}_t)\right] -\mathbb{E}\left[ \min_{x\in \mathcal{K}}\sum_{t=1}^T h_t(x)\right] \\
& = \mathbb{E}\left[\sum_{t=1}^T \tilde{h}_t(\hat{x}_t)\right] -\mathbb{E}\left[ \sum_{t=1}^T h_t(x^*)\right].
\end{align}
The second line follows since $\mathbb{E}[\sum_{t=1}^T \tilde{h}_t(\hat{x}_t)] = \mathbb{E}[\sum_{t=1}^{T-1} \tilde{h}_t(\hat{x}_t)] + \mathbb{E}[\tilde{h}_T(\hat{x}_T)]$ and $\mathbb{E}[\tilde{h}_T(\hat{x}_T)] = \mathbb{E}_{\{(\xi_t,S_t)\}_{t=1}^{T-1}}[ \mathbb{E}_{(\xi_T,S_T)}[\tilde{h}_T(\hat{x}_T)|\{(\xi_t,S_t)\}_{t=1}^{T-1} ]] = \mathbb{E}_{\{(\xi_t,S_t)\}_{t=1}^{T-1}}[  h_T(\hat{x}_T) ] = \mathbb{E}[h_T(\hat{x}_T)]$.  Backwards induction on $T$ yields the desired equality. 



We now bound the absolute difference between $\sum_{t=1}^T h_t(x)$ and $\sum_{t=1}^T \hat{h}_t(x)$ for all $x \in \mathcal{K}$. Since $\tilde{h}_t(x)$ is a random variable, so we seek to bound this absolute difference with probability 1. This will ensure that our bound holds against adaptive adversaries. By the definitions of $h(x)$ and $\hat{h}(x)$ we have,
\begin{align*}
\left|\sum_{t=1}^T \left[ h_t(x) - \tilde{h}_t(x) \right]\right| &= \left|\left(\sum_{t=1}^T \left[\nabla \hat{f}^H_t(x_t) - (\hat{g}_t + Hx_t)\right]\right)^\top x \right| \\
&= \left| \left(\sum_{t=1}^T [\nabla \hat{f}_t (x_t) - \hat{g}_t]\right)^\top x \right|.
\end{align*}

Define $\alpha_t := \nabla \hat{f}_t(x_t) - \hat{g}_t$.  Then we can write, 
\begin{equation}\label{eq.hhtilde}
\left|\sum_{t=1}^T h_t(x) - \tilde{h}_t(x)\right| \leq \|x\|_2 \|\sum_{t=1}^T\alpha_t\|_2.
\end{equation}

We next proceed to bound $\mathbb{E}\left[\|\sum_{t=1}^T \alpha_t\|_2\right]^2$. By Lemma \ref{lemma.alphas} stated below, $\mathbb{E}[\alpha_t^\top \alpha_{t'}] = 0$ for $t\neq t'$.   
\begin{align*}
\mathbb{E}\left[\left\|\sum_{t=1}^T \alpha_t\right\|_2\right]^2 &\leq \mathbb{E}\left[\left\|\sum_{t=1}^T \alpha_t\right\|^2_2\right]\quad \text{ (Jensen's inequality)}\\
&= \sum_{t=1}^T \mathbb{E}\left[\|\alpha_t\|^2_2\right] + 2 \sum_{t<t'}\mathbb{E}\left[\alpha_t^\top \alpha_{t'}\right]\\
&= \sum_{t=1}^T \mathbb{E}\left[\|\nabla \hat{f}_t(x_t)-\hat{g}_t\|^2_2\right]\\
&\leq \sum_{t=1}^T \mathbb{E}\left[2\|\nabla \hat{f}_t(x_t)\|^2_2 + 2 \|\hat{g}_t\|^2_2\right]\\
&\leq 4T \cdot \frac{16M^2n^2}{\gamma}
\end{align*}
where the last line follows 
from Lemma \ref{lem.bound}, and the fact that if $\|\hat{g}_t\|_2\leq G$ then $\|\nabla \hat{f}_t(x_t)\|_2\leq G$ by Jensen's inequality. 

Plugging this bound into Equation \eqref{eq.hhtilde} gives,
\begin{align*}
\left| \sum_{t=1}^T \left[ h_t(x)- \tilde{h}_t(x)\right]\right| \leq D_{\mathcal{K}} \sqrt{\frac{64 M^2 n^2 T}{\gamma}},
\end{align*}
which implies,
\begin{align*}
\mathbb{E}\left[\sum_{t=1}^T h_t(x^*)\right] &\geq \mathbb{E}\left[\sum_{t=1}^T \tilde{h}(x^*)\right] - \frac{8 D_{\mathcal{K}}Mn\sqrt{T} }{\sqrt{\gamma}} \\
& \geq \mathbb{E}\left[ \min_{x\in \mathcal{K}} \sum_{t=1}^T \tilde{h}(x) \right] - \frac{8D_{\mathcal{K}}Mn\sqrt{T} }{\sqrt{\gamma}}.
\end{align*}
Combining this with Equation \eqref{eq.hbound} completes the proof.
\end{proof}

The following lemma was asserted without proof in \cite{ST13}.  We prove it here for completeness.

\begin{lemma}\label{lemma.alphas}
Let $\alpha_t = \nabla \hat{f}_t(x_t) - \hat{g}_t$. Then, for $t<t'$ it holds that $\mathbb{E}[\alpha_t^\top \alpha_{t'}] = 0$, where the expectation is taken over the randomization of the algorithm used to build the estimates of the gradient $\{\hat{g}_t\}_{t=1}^T$.
\end{lemma}

\begin{proof}
\begin{align*}
\mathbb{E}[\alpha_t^\top \alpha_{t'}] &= \mathbb{E}[(\nabla \hat{f}_t(x_t) - \hat{g}_t)^\top (\nabla \hat{f}_{t'}(x_{t'}) - \hat{g}_{t'})]\\
&= \mathbb{E}[\nabla \hat{f}_t(x_t)^\top \nabla \hat{f}_{t'}(x_{t'})] - \mathbb{E}[\nabla \hat{f}_t(x_t)^\top \hat{g}_{t'}] - \mathbb{E}[\nabla \hat{f}_{t'}(x_{t'})^\top \hat{g}_t] + \mathbb{E}[\hat{g}_t^\top \hat{g}_{t'}] \\
& = \nabla \hat{f}_t(x_t)^\top \nabla \hat{f}_{t'}(x_{t'}) - \nabla \hat{f}_t(x_t)^\top \nabla \hat{f}_{t'}(x_{t'}) - \nabla \hat{f}_{t'}(x_{t'})^\top \nabla \hat{f}_t(x_t) + \mathbb{E}[\hat{g}_t^\top \hat{g}_{t'}]
\end{align*}
We now show that $\mathbb{E}[\hat{g}_t^\top \hat{g}_{t'}] = \nabla \hat{f}_{t'}(x_{t'})^\top \nabla \hat{f}_t(x_t)$.
\begin{align*}
\mathbb{E}[\hat{g}_t^\top \hat{g}_{t'}] &= \mathbb{E}_{1,...t'-1}[ \mathbb{E}_{t'}[\hat{g}_t^\top \hat{g}_{t'}|t=1,...t'-1]] \\
& = \mathbb{E}_{1,...t'-1}[ \hat{g}_t^\top \mathbb{E}_{t'}[ \hat{g}_{t'}|t=1,...t'-1]]\\
& = \mathbb{E}_{1,...t'-1}[ \hat{g}_t^\top \nabla\hat{f}_{t'}(x_{t'})]\\
& = \nabla\hat{f}_{t}^\top(x_t) \nabla\hat{f}_{t'}(x_{t'})
\end{align*} 
\end{proof}

\end{document}